\documentclass[journal]{IEEEtran}
\ifCLASSINFOpdf
\else
\fi
\usepackage{cite}
\usepackage{bm}
\usepackage{amsmath,amssymb,amsfonts}
\usepackage{graphicx}
\usepackage{textcomp}
\usepackage{xcolor}
\usepackage{stfloats}
\usepackage{amsmath}
\usepackage{algorithm}
\usepackage{algorithmicx}
\usepackage{algpseudocode}
\usepackage{amsmath,amsfonts}

\newtheorem{proof}{Proof}
\newtheorem{definition}{Definition}
\newtheorem{lemma}{Lemma}
\newtheorem{theorem}{Theorem}

\begin{document}
\bibliographystyle{ieeetr}
%
\title{Vehicular Multi-Tier Distributed Computing with Hybrid THz-RF Transmission in Satellite-Terrestrial Integrated Networks}
%
%
%

\author{Ni~Zhang,~Kunlun~Wang,~Wen~Chen,~Jing~Xu,~Yonghui~Li,~and~Arumugam~Nallanathan
\thanks{N. Zhang, K. Wang and J. Xu are with the Shanghai Key Laboratory of Multidimensional Information Processing, East China Normal University, Shanghai
200241, China, and also with the School of Communication and Electronic Engineering, East China Normal University, Shanghai 200241, China (e-mail: 71255904050@stu.ecnu.edu.cn; klwang@cee.ecnu.edu.cn; jxu@cee.ecnu.edu.cn).}
\thanks{W. Chen is with the Department of Electronic Engineering, Shanghai Jiao Tong University, Shanghai 200240, China (e-mail: wenchen@sjtu.edu.cn).}
\thanks{Y. Li is with the School of Electrical and Information Engineering, The University of Sydney, Australia (e-mail: yonghui.li@sydney.edu.au).}
\thanks{A. Nallanathan is with the School of Electronic Engineering and Computer Science, Queen Mary University of London, E1 4NS London, U.K. (e-mail: a.nallanathan@qmul.ac.uk).}}

\maketitle

\begin{abstract}
In this paper, we propose a  Satellite-Terrestrial Integrated Network (STIN) assisted vehicular multi-tier distributed computing (VMDC) system leveraging hybrid terahertz (THz) and radio frequency (RF) communication technologies. Task offloading for satellite edge computing is enabled by THz communication using the orthogonal frequency division multiple access (OFDMA) technique. For terrestrial edge computing, we employ non-orthogonal multiple access (NOMA) and vehicle clustering to realize task offloading. We formulate a non-convex optimization problem aimed at maximizing computation efficiency by jointly optimizing bandwidth allocation, task allocation, subchannel-vehicle matching and power allocation. To address this non-convex optimization problem, we decompose the original problem into four sub-problems and solve them using an alternating iterative optimization approach. For the subproblem of task allocation, we solve it by linear programming. To solve the subproblem of sub-channel allocation, we exploit many-to-one matching theory to obtain the result. The subproblem of bandwidth allocation of OFDMA and the subproblem of power allocation of NOMA are solved by quadratic transformation method. Finally, the simulation results show that our proposed scheme significantly enhances the computation efficiency of the STIN-based VMDC system compared with the benchmark schemes.
\end{abstract}

\begin{IEEEkeywords}
Satellite-Terrestrial Integrated Networks (STIN), task offloading, orthogonal frequency division multiple access (OFDMA), non-orthogonal multiple access (NOMA), alternating optimization algorithm (AO), many-to-one matching.
\end{IEEEkeywords}

%
\IEEEpeerreviewmaketitle

\section{Introduction}
%
%
%
%
With the widespread adoption of fifth generation (5G) networks and the development of sixth generation (6G) technology, access to more devices, low latency and low energy consumption with high reliability  have become key research priorities \cite{10660910}. Internet of Vehicles (IoV) has significant demands in these areas. IoV is driving the automotive industry toward intelligence by integrating information interaction between vehicles. However, this progress also introduces new challenges.
\par Existing mobile communication networks may not be able to meet the requirements of emerging technologies and applications that require low latency and high efficiency. As a result, the focus has shifted to developing advanced communication network architectures to build integrated network systems across regions, airspace, and sea domains to achieve a truly global network with seamless coverage. Satellite-Terrestrial Integrated Networks (STIN) has attracted significant attention by integrating satellite systems, airborne networks, and terrestrial communications. STIN addresses the limitations of single-network systems, particularly in challenging environments like oceans and mountains, where terrestrial communication systems cannot deliver reliable high-speed wireless access. STIN provides wide coverage, high throughput, flexible deployment. Computation offloading in STIN enhances computation efficiency and reduces energy consumption. Meanwhile, it ensures service quality by optimizing the allocation of computing tasks \cite{10542407} \cite{10498091}. This provides solid technical support for 6G networks to achieve seamless global coverage, high-speed intelligence, and secure communication services \cite{9520341}. Low Earth Orbit (LEO) satellites, with their low orbital altitudes and high speeds, play a critical role in STIN, in expanding communication links and connecting remote regions not covered by terrestrial networks to the global communication system. Meanwhile, they can also provide communication connections for IoT devices distributed all over the world. 
\par Orthogonal frequency division multiple access (OFDMA) technology has become the key transmission technology in terrestrial wireless communication networks \cite{8647301}\cite{9530374}. Its integration into satellite mobile communication systems facilitates the convergence of satellite and terrestrial networks. Terahertz (THz) communication technology with its immense bandwidth, supports ultra-high wireless communication speeds, making it crucial for meeting future communication demands. OFDMA technology enhances spectral efficiency by dispersing data across multiple closely spaced sub-carriers \cite{9399121}. The combination of them enables ultra-high wireless communication speeds especially in the area of satellite communications \cite{9371019} \cite{9347979} .
\par In addition, non-orthogonal multiple access (NOMA) technology allows multiple users to share the same wireless resources by performing successive interference cancellation (SIC) at the receiver to eliminate interference \cite{9850428}. This improves spectrum utilization supports more user connections, alleviates data traffic congestion, and reduces latency \cite{10278913} \cite{9013841}. Thus, NOMA can be utilized in terrestrial task offloading. The joint optimization of task offloading, user clustering, computational resource allocation and transmit power control enhances system task processing efficiency, reduces task processing time and improve reliability and computation performance.


 

\subsection{Related Work}
The task offloading in the STIN system transfers computational tasks from users to edge servers (e.g. satellites, unmanned aerial vehicles (UAVs), high-altitude platforms), which reduces user computation loads and significantly reduces the processing delay of computational tasks\cite{8493149}. Liu \MakeLowercase{\textit{et al.}} \cite{9626560} proposed a wireless power transmission (WPT)-enabled space-air-Ground power Internet of Things (SAG-PIoT) architecture that assigns tasks to local devices, UAV and LEO for computation. Dynamic task offloading and resource scheduling have been investigated and demonstrated the superior performance of task computation offloading. Chai \MakeLowercase{\textit{et al.}} \cite{10024305} developed a model for a joint multi-task mobile edge computing (MEC) system based on UAV-assisted aerial base station(BS). Using a training method that combines the attention mechanism and the proximal policy optimization collaborative (A-PPO) algorithm to train the data, they minimized consumption costs during task offloading. Chen \MakeLowercase{\textit{et al.}}\cite{9552467} considered a satellite-air-ground integrated network (SAGIN) network model that contains multiple LEOs connected to a cloud server, a UAV and a BS. Their distributed robust latency optimization algorithm reduced latency by offloading tasks from UAVs to BSs and cloud servers. Di \MakeLowercase{\textit{et al.}}\cite{8571192} proposed an architecture for terrestrial satellite networks where each small cell assists the macro cell in offloading traffic, thus enabling data offloading efficiently. By optimizing terrestrial data offloading, satellite association, and resource allocation, they maximized total rates and user access.
\par The THz band offers abundant frequency resources and an exceptionally wide operational bandwidth, enabling THz communication systems to support ultra-high data rates\cite{9352550} \cite{9672716} \cite{9629223}. The THz band provides numerous advantages, including high transmission rates, large capacity, enhanced security, and strong anti-interference capabilities. Wu \MakeLowercase{\textit{et al.}} \cite{10681251} proposed a MEC system consisting of a user, an IRS and a UAV, where the user can make a decision to choose to offload the task to the UAV or to execute it locally. Using THz technology, the system increased transmission rates, improved robustness, and reduced energy consumption. Wang \MakeLowercase{\textit{et al.}}\cite{10278635} proposed an intelligent reflecting surface (IRS)-assisted MEC system that solves the UAV energy minimization problem by alternating optimization technique. Yuan \MakeLowercase{\textit{et al.}}\cite{10341311} optimized a SAGIN model considering a multi-band THz/RF channel. The overall performance of SAGIN is improved by optimizing the allocation of THz and RF channels and maximizing the node fairness index (NFI) for multi-band THz and radio frequency (RF) communications in SAGIN. By using THz technique, the spectrum resources can be utilized more efficiently and the overall spectral efficiency of the network can be improved. Yu \MakeLowercase{\textit{et al.}}\cite{9328513} constructed a framework of EC-SAGINs consisting of remote vehicles, LEOs, medium earth orbits (MEOs) and high earth orbits (HEOs). They have addressed the challenges faced by vehicles in remote areas without terrestrial edge computing (TEC) facilities, such as 5G BSs and multiple road side units (RSUs), to access cloud server services. Their model optimized joint offloading decisions and caching strategies while integrating THz technology to enhance network coverage, especially in areas that are geographically isolated or lack the reach of existing terrestrial infrastructure.
\par The use of NOMA for terrestrial communications in STIN allows multiple users to share the same time and frequency resources. It differentiates the signals of users by using SIC \cite{9417280}. This significantly improves the efficiency of spectrum utilization. Wang \MakeLowercase{\textit{et al.}}\cite{9036885} proposed a NOMA-based task scheduling framework that offloads multiple task nodes to multiple nearby auxiliary nodes for execution of their tasks via NOMA. The total cost is minimized by optimizing task scheduling and sub-channel allocation. NOMA technique helps to reduce the delay of task execution by increasing the data transmission rate and reducing the interference between users, making it highly suitable for the large-scale device requirements of in Industrial Internet of Things (IIoT) systems. Sheng \MakeLowercase{\textit{et al.}} \cite{8972932} analyzed the relationship between the offloading latencies of a pair of NOMA users and the mutual co-channel interference they experience. Simulations demonstrated that NOMA effectively reduces average user offloading delays while accommodating more users for task offloading. Ding \MakeLowercase{\textit{et al.}} \cite{9679390} studied a general hybrid  Non-Orthogonal Multiple Access-Mobile Edge Computing (NOMA-MEC) offloading strategy. The proposed strategy provides users with more flexible task offloading opportunities. Tasks can be computed locally by users, and at the same time, they can also be offloaded to BS for computation. It improves the task computing efficiency.  
Xu \MakeLowercase{\textit{et al.}}\cite{9272879} investigated the joint task offloading and resource allocation problem of MECs in NOMA-based heterogeneous networks (HetNets) to achieve minimizing the energy consumption of all users. The application of NOMA technique in HetNets aims to achieve a trend of improving system throughput and spectral efficiency.

\subsection{Contributions}
However, the aforementioned works did not consider the integration of THz and OFDMA in an STIN-based vehicular task offloading system. Therefore, based on previous research, we propose a vehicular multi-tier distributed computing (VMDC) framework in STIN system that utilizes OFDMA in the THz to achieve terrestrial-to-satellite communication. By using OFDMA technology, this wide frequency band can be fully exploited and divided into multiple orthogonal sub-carriers, and each sub-carrier can independently carry data, which greatly increases the total data transmission rate and thus optimizes the performance of the system. Satellite communication in conventional frequency bands often encounters numerous interference sources. The high-frequency properties of THz reduce the likelihood of interference signal generation and propagation, which improves the anti-interference ability of satellite communication.
\par The main contributions of this work are as follows:
\begin{itemize}
\item In this paper, we propose a VMDC framework in STIN system and optimize bandwidth allocation, task allocation, sub-channel matching and power allocation to maximize the computation efficiency, defined as the ratio of the sum of task data sizes to the sum of energy consumption.
\item We propose a hybrid THz-RF transmission scheme for multi-tier distributed task offloading. Using this framework, the closed-form expressions of task transmission delay and energy consumption are obtained through theoretical analysis of the channel model.
\item Based on this model we need to solve the problem of maximizing the computation efficiency of a non-convex problem. We employ an alternating optimization algorithm to decouple it into four subproblems: task allocation, bandwidth allocation in OFDMA-based communication, power allocation in NOMA-based communication, and sub-channel matching. Specifically, the sub-channel matching problem is solved using the many-to-one approach. And we propose a novel algorithm based on the theory of many-to-one two-sided matching.
\item In the simulations, extensive data results are provided to validate the effectiveness of our proposed optimization scheme. Compared with various allocation optimization methods, the proposed scheme demonstrates significantly higher computation efficiency. Moreover, compared to traditional allocation methods, it better meets the needs of individual vehicles, effectively reduces energy consumption, and enhances computation efficiency.
\end{itemize}
\par Notations: In this paper, matrices are denoted by uppercase boldface characters. ${{\mathbb{R}}^{m\times n}}$ and $\mathbb{E}\left ( \cdot  \right )$ denote the complex space of dimension and the expectation. The conjugate-transpose of matrix A is denoted by $A^H$.  

\section{System Model}
\subsection{System Description}
A STIN-assisted VMDC system is shown in Fig. \ref{sat1.eps}, which consists of multiple vehicles, multiple RSUs and a LEO. We denote the set of vehicles by $H=\left \{ H_{1}, H_{2},\cdots, H_{m}\right \}$. The sub-channels from vehicles to RSU are represented by $\left \{ SC_{1}, SC_{2},\cdots, SC_{f}\right \}$. The set of RSUs can be 
denoted as $\left \{U_{1}, U_{2},\cdots, U_{l}\right \}$. There is a single $J$-antenna BS with edge computing server in the proposed VMDC system. The $U$-antenna RSU serves as a relay without computational capability, which increases the transmission distance of the task, and reduces the computational energy consumption and delay.
For terrestrial task offloading, the task is transmitted from vehicle to the RSU via NOMA technique. Then the task is also transmitted from RSU to the BS by the multiple-input multiple-output (MIMO) technique. For satellite task offloading, the vehicle transmits a portion of the task directly to the satellite for computation by the OFDMA technique in the THz band. According to \cite{9328513}, the vehicle can support both THz and RF communications. The vehicle ensures that they do not interfere with each other by having separate antennas. 
\begin{figure}[htbp]
\centerline{\includegraphics[scale=0.12]{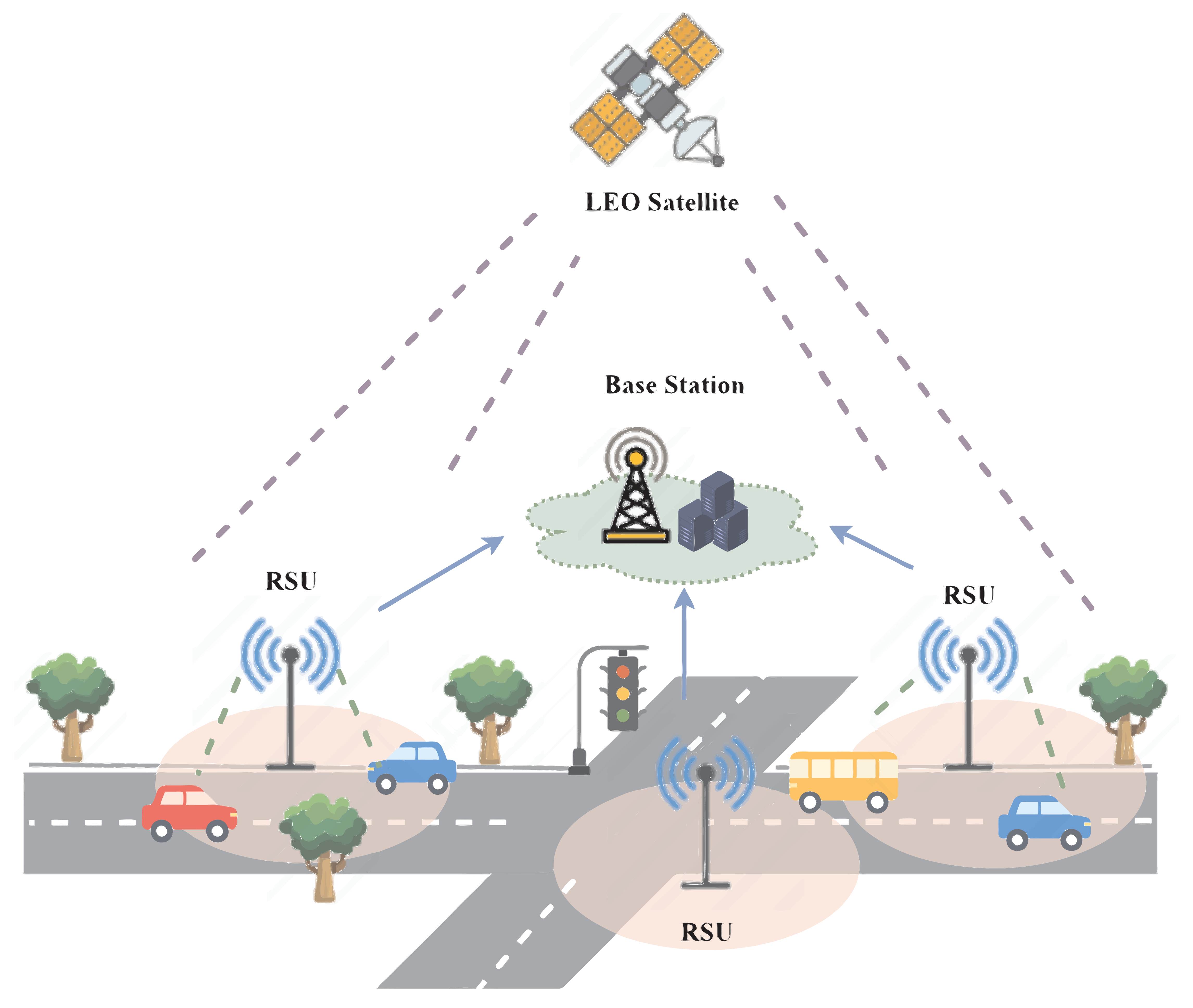}}
\caption{System model.}
\label{sat1.eps}
\end{figure}
\par We divide the vehicles into $N$ clusters with $K$ vehicles in each cluster, $K$ is random and the largest $K$ is $q_{max}$. $H_{n,k}$ is denoted as the $k$th vehicle in the $n$th cluster. In the NOMA transmission process, $\eta_{f,n}\in\left \{ 0,1 \right \} $ denotes the sub-channel matching coefficients to match the $F$ sub-channels with the $N$ cluster vehicles, $\eta_{f,n}=1$ means that the $f$th sub-channel is occupied by the $n$th cluster of vehicles, and $\eta_{f,n}=0$ means that the $f$th sub-channel is not occupied by the $n$th cluster. A cluster can occupy multiple sub-channels, but a sub-channel can be occupied by only one cluster of vehicles. The sub-channels are mutually orthogonal. 
\par Set $r_m$ as the coverage radius of the RSU and $e_m$ as the vertical distance between the RSU and the vehicle. 
$I_{m}$ is the distance exercised by the vehicle through the coverage radius of the RSU. It can be denoted as:
\begin{align}
I_{\rm m}=2\sqrt{r_{\rm m}^{2}-e_{\rm m}^{2} } .\label{1} 
\end{align}
Thus the time spent by the vehicle traveling within the coverage radius of an RSU is denoted as:
\begin{align}
t_{\rm m}^{\rm stay}=\frac{I_{\rm m} }{V_{\rm m}},\label{2} 
\end{align}
where $V_{\rm m}$ denotes the vehicle speed.
\subsection{Communication Model}
\subsubsection{Vehicles-to-Satellite}
In this STIN network, the vehicle communicates with a LEO satellite with THz band for the process of vehicle offloads tasks to the satellite. This satellite can cover all positions within the area. 
 
\par Following previous work \cite{10341311}, the THz communication channel is modeled as follows. Due to the presence of molecular absorption and dense molecular arrangement, line-of-sight (LOS) transmission is much more significant than nonlinear-line-of-sight (NLOS) transmission \cite{9661374}.
The received signal power $P_m^s$ at the LEO satellite from the vehicle $m$ is given by \cite{10341311}
\begin{align}
P_{\rm m}^{\rm s} =\rho P_{\rm m}^{\rm H} \left ( \frac{c}{4\pi f_{\rm m}^{\rm s} } \right ) ^{2} \chi _{s} G_{\rm m}^{\rm s} G_{\rm m}^{\rm v} e^{-R_{\rm a}\left ( f_{\rm m}  \right ) d_{\rm 0}  } d_{\rm 0}^{-2 } ,\label{3}
\end{align}
where $P_{m}^{H}$, $\rho$, $f_{m}^{s}$, $\chi_{s}$, $R_{a}\left ( f_{m}  \right )$ and $d_{0}$  respectively represent the transmission power of the vehicle, the channel power gain at the reference distance $d_0=1$ m, the frequency used by the vehicle, small-scale fading, molecular absorption coefficient and distance from vehicle to satellite. $G_{m}^{s}$ and $G_{m}^{v}$ represent beam-forming gains of the main and side lobes \cite{9119462}, respectively.
\par The received SINR at the satellite is denoted as:
\begin{align}
SINR_{\rm m} ^{\rm s}=\frac{P_{\rm m} ^{\rm s}}{H^{\rm s}} ,\label{4}
\end{align}
where $H^{s}$ represents the noise power at the LEO satellite. 
\par Vehicles offload tasks to the LEO satellite by OFDMA technology, which allows multiple vehicles to communicate with the satellite at the same time. OFDMA divides the entire frequency band so that multiple vehicles use different orthogonal sub-carriers, and there is no signal interference between different vehicles. We denote $\alpha_m$ as the subcarrier pairing result between the vehicle and subcarrier, $\alpha_{m}$ is generally a binary variable that can only take values of 0 and 1. As the number of subcarriers increases, the system is able to allocate the bandwidth to each vehicles more accurately, which allows $\alpha_{m}$ to take a range of values closer to continuous values between 0 and 1 \cite{8438896}. Therefore, the task transmission rate from the vehicle to the satellite is given by
\begin{align}
R_{\rm m}^{\rm s} =\alpha _{\rm m}B_{\rm s}\log_{2}{\left ( 1+ SINR_{\rm m} ^{\rm s}\right ) } ,\label{5}
\end{align}
where $B_{s}$ is the total bandwidth for the transmission from vehicles to satellite.
\subsubsection{Vehicles-to-RSUs}
For terrestrial VEC, the transmission of the task from the vehicle to the RSU and from the RSU to the BS follows a Rayleigh distribution by \eqref{3}, and $h_{n,k}$ is set to be the channel gain, which is denoted as:
\begin{align}
 h_{\rm n,k} =U_{\rm n,k}\sqrt{{\rm L}_{\rm n,k}   },\label{6} 
\end{align}
where $U_{n,k}$ denotes the small-scale fading parameter and $L_{n,k}$ denotes the large-scale fading parameter \cite{9827408}.
Set the transmit signal $S_{n,k}$ by the vehicle, $y$ is denoted as the signal received by the RSU
\begin{align}
 y=\sum_{\rm k=1}^{\rm K} h_{\rm n,k} \sqrt{P_ {\rm n,k}} S_{\rm n,k} +w ,\label{7} 
\end{align}
where $P_{n,k}$ is the transmission power of $H_{n,k}$.
\par When a vehicle offloads a task to the RSU via the NOMA technique, there is intra-cluster interference between the vehicles. The sub-channels are orthogonal to each other, so there is no cross-cell interference interference. We assume that the channel gains satisfy: $h_{\rm n,1} > h_{\rm n,2} > \cdots >  h_{\rm n,k} $. And SIC technique is applied to the RSU to decode the signals sent by vehicles occupying the same cluster in the same sub-channel. Therefore, the SINR of this transmission process can be expressed as:
\begin{align}
SINR_{\rm n,k}=\frac{P_{\rm n,k}\left | h_{\rm n,k}  \right | ^{2} d_{\rm nR,m}^{-\rho ^{'} }  }{\sum_{k+1}^{K} \left | h_{\rm n,k}  \right | ^{2}P_{\rm n,k} d_{\rm nR,m}^{-\rm \rho ^{'} } +\sigma ^{2}_{\rm n,k} } ,\label{8} 
\end{align}
where $\sigma ^{2}_{\rm n,k}$, $h_{\rm n,k}$ and $d_{\rm nR,m}$ respectively represent the noise power, small-scale fading with $h_{\rm n,k}\sim \mathcal{CN} \left ( 0,1 \right )$ and the distance from vehicle to RSU. And $d_{\rm nR,m}$ can be expressed as:
\begin{align}
d_{\rm nR,m} =\sqrt{a^{2}+\left ( \frac{I_{\rm m} }{2}-V_m{\rm t } \right ) ^{2}  } .\label{9} 
\end{align}
\par The transmission rate of the $k$-th vehicle in the $n$-th cluster is $R_{\rm n,k}$ given by
\begin{align}
R_{\rm n,k} =\sum_{\rm f\in F}w_{\rm f,k} \eta _{\rm f,n} \log_{2} \left ( 1+SINR_{\rm n,k} \right ) .\label{10} 
\end{align}
\par The vehicle is moving during the NOMA transmission, so the transmission channel is also changing dynamically. To make the calculation more accurate, the average rate is used, which can be expressed as:
\begin{align}
\overline{R} _{\rm n,k}=\frac{\int_{0}^{t_{\rm m}^{stay}}R_{\rm n,k}dt }{t_{\rm m}^{\rm stay} }.\label{11} 
\end{align}
\subsubsection{RSUs-to-BS}
RSU acts as a relay during terrestrial transmission from RSU to BS. The transmit signal is set as $\boldsymbol{S}=\left [ S_1,\dots, S_m,\dots,S_M\right ]^T$. 
Therefore, the signal received at the RSU is given by 
\begin{align}
\boldsymbol{y}=\boldsymbol{HBS}+\boldsymbol{n_{\rm 0}},
\end{align}
where $\boldsymbol{B}=\left [ \boldsymbol{b_{\rm 1}},\dots,\boldsymbol{b_{\rm u}},\dots,\boldsymbol{b_{\rm U}}\right ]^H $, $\boldsymbol{H}=\left [ \boldsymbol{h_{\rm 1}},\dots,\boldsymbol{h_{\rm j}},\dots,\boldsymbol{h_{\rm J}}\right ] ^H$ and $\boldsymbol{n_{\rm 0}}=\left [ \boldsymbol{n_{\rm 1}},\dots,\boldsymbol{n_{\rm j}},\dots,\boldsymbol{n_{\rm J}}\right ]^H $ are the beamforming matrix of $U\times M$, the channel matrix of $J\times U$ and channel noise of $J\times 1$ respectively. $\boldsymbol{b_{\rm u}}$ and $\boldsymbol{h_{\rm j}}$ are the rows of $\boldsymbol{B}$ and $\boldsymbol{H}$ respectively.
\par The received SINR can be expressed as:
\begin{align}
SINR_{m}=\frac{P_{\rm R}\left | \boldsymbol{h}^H_{\rm m} \boldsymbol{b}_{\rm m} \right | ^2}{\sum_{i\ne m}^{M} P_{R}\left |\boldsymbol{h} ^H_{\rm i} \boldsymbol{b}_{\rm i}\right |^2+ \boldsymbol{n}^2},
\end{align}
where $P_{\rm R}$ and $\boldsymbol{n}^2$ respectively represent the RSU's transmit power and the noise power. \(\boldsymbol{b}_m\) and $\boldsymbol{h}_m$ respectively represent \(m\)-th row of the beamforming matrix \(\boldsymbol{B}\) 
and the \(u\)-th row of the channel matrix \(\boldsymbol{H}\).
In order to simplify the analysis, we consider obtaining the precoding matrix B by the traditional optimization method \cite{9305793}. For example, we can design the beamforming matrix to minimize the mean square error (MMSE) between the received signal and the desired signal, which is applicable to multi-user MIMO systems \cite{9417429} \cite{8782618}. Therefore, the transmission rate at which tasks are offloaded from the RSU to the BS is given by 
\begin{align}
R_{\rm m}^{\rm RSU,BS} =\mathbb{E}\left \{ B_{\rm R}\log_{2}{\left ( 1+ SINR_{\rm m}  \right ) } \right \},\label{14}
\end{align}
where $B_{\rm R}$ is the bandwidth for the transmission from RSU to BS.
\subsection{Computation Model}
We set the task transmitted by each vehicle as $R_{m} =\left ( L_{m},C_{m} ,T_{m}^{\max} \right )$, where $L_{m}$ represents the total number of CPU cycles required to perform task computation, where $C_{m}$ represents the size of task data and $T_{m}^{\max}$ represents the maximum time allowed for processing tasks. 
\subsubsection{Local Computing}
Vehicles have computation capabilities. We set $\theta _{\rm m}$ and $\zeta _{\rm m}$  as the task allocation coefficient,  which are continuous variable. They determine how much task is allocated to the BS and the LEO satellite. $\theta_{\rm m}L$ is the amount of tasks handled by the BS, and $\zeta_{\rm m}L$ is the amount of the task to be handled by the satellite, so the amount of the task to be handled locally is $\left(1 - \zeta_{\rm m} - \theta_{\rm m}\right)L_{\rm m}$. Therefore, 
the local computation delay is expressed as:
\begin{align}
T_{\rm m}^{\rm loc}=\frac{\left ( 1-\theta _{\rm m}-\zeta  _{\rm m}  \right ) L_{\rm m}}{Z _{\rm m}} ,\label{15}
\end{align}
and the local computation energy consumption is expressed as:
\begin{align}
E_{\rm m}^{\rm loc} =\left ( 1-\theta _{\rm m} - \zeta  _{\rm m}\right ) L_{\rm m} \varphi _{\rm m}^{\rm loc} \left ( Z_{\rm m }\right ) ^{2} ,\label{16}
\end{align}
where $\varphi _{\rm m}^{\rm {loc}}$ denotes the processing energy
coefficient for the vehicle CPU. And $Z _{\rm m}$ represents the vehicle's CPU frequency. 
\subsubsection{Edge Computing}
The vehicle has computation capability. The vehicle is allocated the computation task of $\left ( 1-\theta _{\rm m}-\zeta  _{\rm m}  \right )C_m$, the satellite is allocated task of $\zeta _{\rm m}C_{\rm m}$ and the BS is allocated task of $\theta _{\rm m}C_{\rm m}$. Therefore the transmission delay $T_{\rm m}^{\rm H,RSU}$ for vehicle offloading task to RSU is denoted as:
\begin{align}
T_{\rm m}^{\rm H,RSU} =\frac{\theta _{\rm m} C_{\rm m}}{\overline{R}_{\rm n,k} }.\label{17}
\end{align}
And the transmission delay of the task from the RSU to the BS is denoted as:
\begin{align}
T_{\rm m}^{\rm RSU,BS} =\frac{\theta _{\rm m}C_{\rm m}}{R_{\rm m}^{\rm RSU,BS} }  .\label{18}
\end{align}
Since the RSU as a relay, it does not have computational capability. The delay consumed by the task to offload from the vehicle to BS for computation is denoted as:
\begin{align}
T_{\rm m}^{\rm BS} =\frac{\theta _{\rm m}L_{\rm m}}{Z_{\rm B} }.\label{19}
\end{align}
The task offloading delay from the vehicle to the LEO satellite is given by:
\begin{align}
T_{\rm m}^{\rm H,sat} =\frac{\zeta _{\rm m}C_{\rm m}}{R_{\rm m}^{\rm S} }.\label{20}
\end{align}
The computation delay by the LEO satellite's computation is denoted as:
\begin{align}
T_{\rm m}^{\rm sat} =\frac{\zeta _{\rm m}L_{\rm m}}{Z_{\rm S} } .\label{21}
\end{align}
Thus, total delay by ground task offloading is given by 
\begin{align}
T_{\rm m}^{\rm off}=T_{\rm m}^{\rm H,RSU}+T_{\rm m}^{\rm RSU,BS}+T_{\rm m}^{\rm BS} .\label{22}
\end{align}
\par On the other hand, the energy consumed by the vehicle to offload the task to the RSU is expressed as:
\begin{align}
E_{\rm m}^{\rm up1} =P_{\rm m}^{\rm H} T_{\rm m}^{\rm H,RSU},\label{23}
\end{align}
where $P_{\rm m}^{\rm H}$ represents the vehicle's transmit power. The energy consumed by the transmission of the task from the RSU to the BS is expressed by 
\begin{align}
E_{\rm m}^{\rm up2} =P_{\rm R} T_{\rm m}^{\rm RSU,BS}.\label{24}
\end{align}
The energy consumed by the total number of CPU cycles of $\theta _{\rm m}L_{\rm m}$ to perform the computation at BS is denoted by 
\begin{align}
E_{\rm m}^{\rm BS} =\theta _{\rm m}L_{\rm m}\varphi _{\rm m}^{\rm B} \left (Z_{\rm B}  \right ) ^{\rm 2},\label{25}
\end{align}
where $Z_{\rm B}$ represents the BS’s CPU frequency. The STIN-assisted VMDC system network also contains the transmission energy consumption of the vehicles to the LEO satellite and  the LEO satellite's computation energy consumption. For vehicle to satellite transmission, the vehicle energy consumption is denoted by 
\begin{align}
E_{\rm m}^{\rm H,sat} =P_{\rm m}^{\rm H} T_{\rm m}^{\rm H,sat} .\label{26}
\end{align}
The computation energy consumption at the LEO satellite is expressed as:
\begin{align}
E_{\rm m}^{\rm sat} =\zeta _{\rm m} L _{\rm m} \varphi _{\rm m}^{\rm S} \left (Z_{\rm S}  \right ) ^{\rm 2},\label{27}
\end{align}
where $Z_{\rm S}$ represents the satellite’s CPU frequency. 
\par In all, the total energy consumption of the vehicle performing task offloading is given by
\begin{align}
E_{\rm m}=E_{\rm m}^{\rm loc}+E_{\rm m}^{\rm up1}+E_{\rm m}^{\rm up2}+E_{\rm m}^{\rm BS}+E_{\rm m}^{\rm H,sat}+E_{\rm m}^{\rm sat}.\label{28}
\end{align}
\par We introduce computation efficiency, denoted by $\mathbb{J}$. It is defined as the ratio of the number of bits of the task to the total energy consumption. According to (\ref{28}), the optimization objective is then formulated as:
\begin{align}
\mathbb{J} =\frac{\sum_{\rm m=1}^{\rm M}C_{\rm m}  }{\sum_{\rm m=1}^{\rm M} E_{\rm m}} .\label{29}
\end{align}
\section{Problem Formulation and Analysis}
\par In this section, we optimize the computation efficiency for our proposed STIN-assisted VMDC system. The computation efficiency reflects the delay and energy consumption of task offloading. The computation efficiency is maximized by optimizing the bandwidth allocation coefficient $\alpha_{\rm m}$ of OFDMA, power allocation coefficient $P_{\rm m}$ and sub-channel matching coefficient $\eta _{f,n}$ of NOMA and task allocation coefficient $\theta_{\rm m}$ and $\zeta_{\rm m}$.

\begin{subequations}
\begin{align}
\label{30}
\textbf{P0}:&\mathop{\max}\limits_{ \bm{\eta},\bm{P},\bm{\alpha},\bm{\theta},\bm{\zeta} }\mathbb{J}
&\begin{array}{r@{\quad}l@{}l@{\quad}l}
\end{array}
\end{align}
\begin{align}
\rm{s.t.} \quad& 0\le \theta _{\rm m} \le 1, \forall m,\\
&0\le \zeta_{\rm m}  \le 1, \forall m,\\
&0\le \alpha _{\rm m} \le 1, \forall m,\\
&\eta_{\rm f,n}\in\left \{ 0,1 \right \},\forall f,n,\\
&0\le1-\theta _{\rm m} -\zeta _{\rm m} \le1, \forall m,\\
&\sum_{n\in  N}^{} \eta _{f,n} =1,\forall f,n,\\
&\sum_{\rm m=1}^{\rm M} \alpha _{\rm m} = 1, \forall m,\\
&0\le P_{\rm m} \le P^{\rm max}, \forall m,\\
&T_{\rm m}^{\rm H,RSU} \le t_{\rm m}^{\rm stay}, \forall m,\\
&0< k< q_{\rm max},\\
&R_{\rm m}^{\rm s}\ge\frac{\xi_{\rm m}C_{\rm m}  }{T_{\rm m}^{\rm max}-T_{\rm m}^{\rm sat}  }, \forall m,\\
&R_{\rm n,k} \ge\frac{\theta_{\rm m}C_{\rm m}  }{T_{\rm m}^{\rm max}-T_{\rm m}^{\rm BS}  }, \forall m,n,k,
\end{align}
\begin{align}
&\max\left \{T _{\rm m}^{\rm loc},T_{\rm m}^{\rm off}, T_{\rm m}^{\rm H,sat}+T_{\rm m}^{\rm sat}\right\} \le T_{\rm m}^{\rm max}, \forall m.
\end{align}\label{30n}
\end{subequations}

\par Constraint 30(b), 30(c) and 30(d) restrict continuous variables $\theta_{\rm m}$, $\zeta_{\rm m}$ and $\alpha_{\rm m}$ to values between 0 and 1. The sub-channel matching variable is a binary variable that can only take values of 0 and 1 in 30(e). Constraint 30(f), the $\left(1 - \zeta_{\rm m} - \theta_{\rm m}\right)$ denotes the proportion of tasks allocation locally, which is guaranteed the range of 0 and 1. Constraint 30(g) restricts the number of clusters allowed to communicate on a sub-channel. 
The $SC_{\rm f}$ is only occupied by one cluster. Constraint 30(h) restricts the range of values for bandwidth allocation, the sum of the bandwidth allocation coefficients is 1. 30(i) limits the transmit power of vehicle $P_m$. 30(j) indicates that the vehicle complete the task offloading in $ t_{\rm m}^{\rm stay}$. Constraint 30(k) means that the number of vehicles in each cluster does not exceed $q_{\rm max}$. Constraint 30(l) and 30(m) limit the transmission rate of vehicle offloading. Constraint 30(n) limits the total delay by ground task offloading and satellite task offloading to a maximum of $T_{\rm max}$. 
\par Since the problem we proposed is a non-convex fractional optimization problem. It is coupled, and thus we adopt the alternating optimization method to address it. The formulated original problem P0 can be decoupled into four sub-problems. The task allocation sub-problem is a linear programming problem that can be easily solved. The bandwidth allocation sub-problem with OFDMA and the power allocation sub-problem with NOMA are fractional problems. We utilize the quadratic transformation method to solve them. For the sub-channel matching with NOMA, we adopt the many-to-one method for solution. Finally, these solutions are iteratively optimized to obtain the optimal computation efficiency of the STIN-assisted VMDC system.
\section{THz Bandwidth Allocation and Task Allocation}
\subsection{Optimal Task Allocation $\theta_{\rm m}$ and $\zeta_{\rm m}$ with Fixed Given $\alpha_{\rm m}$, $\eta_{\rm f,n}$ and $P_{\rm m}$}
\indent The task allocation sub-problem P1 can be expressed as:
\begin{equation}
\begin{aligned}
\label{31}
&\textbf{P1}:\quad \mathop{\max}
\limits_{\bm{\theta},\bm{\zeta}}\mathbb{J}(\bm{\theta_{m},\zeta_{m}})\\
& \begin{array}{r@{\quad}l@{}l@{\quad}l}
\rm{s.t.}&30(b), 30(c), 30(f), 30(j), 30(l), 30(m), 30(n).
\end{array}
\end{aligned}
\end{equation}
\par This subproblem focuses on a linear programming problem where the objective is to maximize computation efficiency while satisfying a set of linear inequalities. Under the constraints, the optimal solution of $\zeta _{\rm m}$ and $\theta_{\rm m}$ are
\begin{align}
\zeta_{\rm m} =\frac{Z}{G_{\rm m}} -\theta_{\rm m}\frac{F_{\rm m}}{G_{\rm m}}-\frac{\phi_{\rm m}}{G_{\rm m}},\label{32}
\end{align}
\begin{align}
\theta_{\rm m} =\frac{Z}{F_{\rm m}}-\zeta _{\rm m}\frac{G_{\rm m}}{F_{\rm m}}-\frac{\phi_{\rm m}}{F_{\rm m}},\label{33}
\end{align}
\begin{align}
\phi_{\rm m} =L_{\rm m}\varphi _{\rm m}^{\rm loc}\left (  Z_{\rm m}\right ) ^2,\label{34}
\end{align}
where $Z$ is introduced by intercept (the point at which a line crosses the y-axis).
And $G_m$ and $F_m$ can be expressed as: 
\begin{equation}
\begin{aligned}
G_{\rm m}=L_{\rm m} \varphi^{s}_{\rm m} Z_{\rm S}^2 -L_{\rm m} \varphi^{\rm loc}_{\rm m} Z_{\rm m}^2 + P_{\rm m}^{\rm H}\frac{C_{\rm m}}{R_{\rm m}^{\rm S}},\label{35}
\end{aligned}
\end{equation}
\begin{equation}
\begin{aligned}
F_{\rm m}=L_{\rm m}\varphi_{\rm m}^{\rm loc}\left(Z_{\rm m}\right)^2+L_{\rm m}\varphi_{\rm m}^{\rm B}\left(Z_{\rm B}\right)^2+ P_{\rm m}^{\rm H}\frac{C_{\rm m}}{\bar{R_{\rm m}}}+\frac{P_{\rm R}C_{\rm m}}{R_{\rm m}^{\rm RSU,BS}}.\label{36}
\end{aligned}
\end{equation}

\subsection{Optimal Bandwidth Allocation $\alpha_{m}$ with Fixed Given $\theta_{m}$, $\zeta_{m}$, $\eta_{f,n}$ and $P_{m}$}
\indent The task allocation sub-problem P2 can be expressed as:
\begin{equation}
\begin{aligned}
\label{37}
&\textbf{P2}:\quad \mathop{\max}
\limits_{\bm{\alpha}}\mathbb{J}(\bm{\alpha_{m}})\\
& \begin{array}{r@{\quad}l@{}l@{\quad}l}
\rm{s.t.} &30(d), 30(h), 30(m), 30(n).
\end{array}
\end{aligned}
\end{equation}
\par The problem is a non-convex fractional problem, so we transform this problem into a linear problem that is easy to solve by quadratic transformation. According to (\ref{29}), sub-problem P2 can be transformed into the following problem by quadratic transformation \cite{boyd2004convex}:
\begin{equation}
\begin{aligned}
\mathop{\max}\limits_{\bm{\alpha, y}}&\left[2y_{\rm m}\sqrt{\sum_{\rm m=1}^{\rm M}C_{\rm m} } -y_{\rm m}^2 \sum_{\rm m=1}^{\rm M} \left ( \frac{H_{\rm m}}{\alpha_{\rm m}}+\sigma_{\rm m} \right )\right],\label{38}
\end{aligned}
\end{equation}
where $y_m$ is introduced by the quadratic transform. And $H_m$ and $\sigma_{\rm m}$ can be expressed as: 
\begin{align}
H_{\rm m}=\frac{\zeta _{\rm m}C_{\rm m}P_{\rm m}^{\rm H} }{B_{\rm s}\log_{2}{\left ( 1+SINR_{\rm m}^{\rm s} \right ) } } ,\label{39}
\end{align}
\begin{equation}
\begin{aligned}
\sigma_{\rm m}=E_{\rm m}^{\rm loc}+E_{\rm m}^{\rm up1}+E_{\rm m}^{\rm up2}+E_{\rm m}^{\rm BS}+E_{\rm m}^{\rm sat}.\label{40}
\end{aligned}
\end{equation}
\par When $\alpha_{\rm m}$ is fixed, take the partial derivative of (\ref38) with respect to $y_{\rm m}$ and set it equal to $0$, we can obtain the optimized solution. We set $y_{\rm m}^{\ast}$ and $\alpha_m^{\ast}$ be the expressions for the optimal solution of (\ref{38}). $y^{\ast}_{\rm m}$ is a function of $\alpha_m^{\ast}$, which can be represented as:
\begin{align}
y^{\ast}_{\rm m}=\frac{\sqrt{\sum_{\rm m=1}^{\rm M} C_{\rm n}} }{\sum_{\rm m=1}^{\rm M}\left ( \frac{H_{\rm m}}{\alpha_m^{\ast}}+\sigma_{\rm m} \right )  } .\label{41}
\end{align}
\par For fixed $y_m$, the objective function of problem (\ref{38}) is a concave function with respect to $\alpha_{\rm m}$, which can be solved by the Convex Optimization Toolbox(CVX). Finally, a converged optimization solution $\alpha_m^{\ast}$ can be obtained by alternatively optimizing $\alpha_{\rm m}$ and $y_m$.
\subsection{Optimal Power Allocation $P_{m}$ with Fixed Given $\alpha_{m}$, $\eta_{f,n}$, $\theta_{m}$ and $\zeta_{m}$}
\indent The power allocation sub-problem P3 can be expressed as:
\begin{equation}
\begin{aligned}
\label{42}
&\textbf{P3}:\quad \mathop{\max}
\limits_{\bm{P}}\mathbb{J}(\bm{P_m})\\
& \begin{array}{r@{\quad}l@{}l@{\quad}l}
\rm{s.t.} & 30(i), 30(j), 30(m), 30(n).
\end{array}
\end{aligned}
\end{equation}
\par According to Shannon Bound, we simplify $\overline{R}_{\rm m}$ due to the complexity and difficulty of the computation. 
\begin{theorem}
$\overline{R} _{\rm m}$ is expressed as:
\begin{align}
\overline{R} _{\rm m}=\frac{1.44P_{\rm m}\int_{0}^{t_{\rm m}^{\rm stay}} \sum_{f\in F} w_{\rm f,k} \eta _{\rm f,n}\frac{\left | h_{\rm n,k}  \right | ^{\rm 2}d_{\rm nR,m}^{- \rm \rho ^{'} } }{n_0}dt }{t_{\rm m}^{\rm stay} } .\label{43}
\end{align}
\end{theorem}
\begin{proof}
Refer to Appendix A.$\hfill\blacksquare$ 
\end{proof}
\par An approximate solution with respect to $P_m$ is obtained. According to (\ref{29}), $\mathbb{J} =\frac{\sum_{m=1}^{M}C_{\rm m} }{\frac{\Phi_m}{P_m} +\psi_m } $ can be transformed by quadratic transformation as \cite{boyd2004convex}:
\begin{equation}
\begin{aligned}
\mathop{\max} \limits_{\bm{P}, {\bm x}}&\left[2x_m\sqrt{\sum_{m=1}^{M}C_m } -x_m^2 \sum_{m=1}^{M}\left ( \frac{\Phi_m}{P_m} +\psi_m \right ) \right],\label{44}
\end{aligned}
\end{equation}
where $x_m$ is introduced by the quadratic transform. And $\Phi_{\rm m}$ and $\psi_{\rm m}$ can be expressed as: 
\begin{align}
\Phi_{\rm m}=\frac{\theta _{\rm m}C_{\rm m}t_{\rm m}^{\rm stay}P_{\rm m}^{\rm H }}{1.44\int_{0}^{t_{\rm m}^{\rm stay}} \sum_{f\in F}w_{ f,k}\eta _{ f,n} \frac{\left | h_{ n,k} \right |^{2}d_{ nR,m}^{-\rm \rho ^{'} }}{n_ 0}dt }  ,\label{45}
\end{align}
\begin{equation}
\begin{aligned}
\psi_{\rm m}=E_{\rm m}^{\rm loc}+E_{\rm m}^{\rm H,sat}+E_{\rm m}^{\rm up2}+E_{\rm m}^{\rm BS}+E_{\rm m}^{\rm sat}   .\label{46}
\end{aligned}
\end{equation}
Similarly, we set $x_m^{\ast}$ and $P_m^{\ast}$ to be the expressions for the optimal solution of (\ref{42}). $x^{\ast}_{\rm m}$ is a function of $P_m^{\ast}$, which can be represented as:
\begin{align}
x^{\ast}_{\rm m}=\frac{\sqrt{\sum_{\rm m=1}^{\rm M} C_{\rm m}} }{\sum_{\rm m=1}^{M}\left ( \frac{\Phi_{\rm m}}{P_{\rm m}^{\ast}}+\psi_{\rm m} \right ) } .\label{47}
\end{align}
\par For fixed $x_m$, the objective function of problem (\ref{44}) is a concave function with respect to $P_m$, which can be solved by CVX. By alternately optimizing $x_m$ and $P_m$, the convergent optimization solution $P_m^{\ast}$ can be acquired.
\section{Many-to-one Matching for Sub-channel Allocation}
To solve the problem of sub-channel matching, we use the matching theory to obtain the optimal sub-channel matching $\eta_{f,n}$ by given the independent variables of power allocation $P_{m}$, bandwidth allocation $\alpha_{m}$, and task allocation $\theta_{m}$ and $\zeta_{m}$.
\indent The sub-channel matching sub-problem P4 can be expressed as:
\begin{equation}
\begin{aligned}
\label{48}
&\textbf{P4}:\quad \mathop{\max}
\limits_{\bm{\eta}}\mathbb{J}(\bm{\eta})\\
& \begin{array}{r@{\quad}l@{}l@{\quad}l}
\rm{s.t.} & 30(e), 30(g), 30(j), 30(m), 30(n).
\end{array}
\end{aligned}
\end{equation}
\indent In this matching problem, cluster and sub-channels are set as agents, denoted by sets $\mathcal{I} =\left \{ \bm{I}_{1}, \bm{I}_{2}, \cdots, \bm{I}_{\rm N} \right \}$  and $\mathcal{F} =\left \{ F_{1}, F_{2}, \cdots, F_{\rm F} \right \}$, respectively, which are selfish and rational. The NOMA is used in the communication process of vehicle to RSU. In this process $M$ vehicles are matched with $F$ sub-channels. The problem of sub-channel matching is solved using a many-to-one matching algorithm \cite{8972932} \cite{9272879} \cite{8014491} . In order to  utilize the sub-channel resources vividly and obtain optimal computation efficiency, the sub-channels ${\rm SC}_{f}$ and vehicles $M_{m}$ are paired with each other to form stable matching pairs. During the matching process, the players have different selection order for another set of players, we denote the sequence of player's preferences by $P=\left \{ P\left ( \bm{I}_{1}  \right ) , \cdots , P\left ( \bm{I}_{N}  \right ), P\left ( SC_{1}  \right ) , \cdots,\right.
\left.P\left ( SC_{F}  \right )\right \}$.
\begin{definition}
Given two disjoint sets $\mathcal{I}$ and $\mathcal{F}$, 
\end{definition}
\begin{enumerate}
\item[1).] $\mu \left ( \bm{I}_{n}  \right ) \subseteq \mathcal{F}$, $\mu \left ( SC_{f}  \right ) \subseteq \mathcal{I}$;
\item[2).] $\left | \mu \left ( SC_{f}  \right )  \right |=1 $;
\item[3).] $\left | \mu \left ( \bm{I}_{n}  \right )  \right | \le q_{\max} $
\item[4).] $\bm{I}_{n} \subseteq \mu \left ( SC _{f} \right ) \Leftrightarrow SC_{f} \subseteq \mu \left ( \bm{I}_{n} \right ) $;
\item[5).] $\left (\mathrm{transitive}\right )$ if $B\succeq_{I_{n}}B^{'}$ and $B^{'}\succ_{I_{n}}B^{''}$, then $B\succeq_{I_{n}}B^{''}$. 
\end{enumerate}
\par For each vehicle, the computation efficiency in the offloading process is the optimization objective, so the computation efficiency of each cluster $I_{n}$ on the sub-channel $SC_{f}$ set to the utility function can be expressed as:
\begin{align}
U_{\rm n} \left ( f \right ) =\mathbb{J}_{\rm m\in M_{n} }.\label{49}
\end{align}
\begin{definition}
The $M$ vehicles are devided into $N$ clusters and each cluster can occupy one or more sub-channels. The preference relation can be expressed as:
\end{definition}
\begin{align}
{SC}_{\rm f}\succ_{\bm{I}_n}{SC}_{f\prime}\Leftrightarrow U_{\rm n} \left ( f \right )> U_{\rm n} \left ( f ^{'}\right ).\label{50}
\end{align}
\begin{definition}
Each sub-channel is allowed to be occupied by only one cluster and there are $K$ vehicles per cluster. The value of $K$ is taken between $h_{1}$ and $h_{2}$. The preference relation can be expressed as:
\end{definition}
\begin{align}
\bm{I}_{\rm k}\succ_{{SC}_{\rm f}}\bm{I}_{k\prime}\Leftrightarrow U_{\rm f} \left ( n \right )> U_{\rm f} \left ( n ^{'}\right ).\label{51}
\end{align}
\indent With different preference settings for the two sets, we can solve the optimization problem of sub-channels matching by using them in a many-to-one two-sided matching theory to make the objective function optimal.
\begin{definition}
Given two disjoint sets $\mathcal{I}$ and $\mathcal{F}$, a matching $\mu$, a pair $\left(I_{\rm n},{\rm SC}_{\rm f}\right)$ is a blocking pair which is $I_{\rm n}\not\in \mu \left ( SC_{\rm f} \right )$ and $SC_{rm f} \not\in \mu \left ( I_{\rm n} \right ) $. It must satisfy
\end{definition}
\begin{enumerate}
\item[1).] $A\succ_{SC_{f} } \mu \left ( SC_{\rm f}  \right )$, $A\subseteq \left \{ I_{\rm n}  \right \} \cup \left \{  SC_{\rm f} \right \} $ and $I_{n} \in A$.
\item[2).] $SC_{\rm f} \succ_{I_{n} } SC_{\rm l}, SC_{\rm l} \in \mu \left ( I_{n} \right ) $.
\end{enumerate}
\begin{definition}
There are no blocking pairs in the matching $\mu$ and this matching $\mu$ is stable.
\end{definition}
\begin{lemma}
If the matching process converges to a match $\mu^{*}$ in VSMA, then $\mu^{*}$ is a stable match.
\end{lemma}
\begin{proof}
Refer to Appendix B.$\hfill\blacksquare$ 
\end{proof}
\begin{theorem}
In VSMA, each cluster applys for matching, and the optimal matching result is reached after finite iterations by the matching $\mu^{\ast}$.
\end{theorem}
\begin{proof}
Refer to Appendix C.$\hfill\blacksquare$ 
\end{proof}

\begin{algorithm}[h]
  \caption{Vehicles and Sub-channels Matching Algorithm(VSMA)} 
  \begin{algorithmic}[1]
    \Require\\
      $\bm{M}$: initial a matching list, which is empty;\\
      $\bm{SC_f}$ and $\bm{I_n}$: preference lists for sub-channels and clusters;\\
      nonempty: a row in the $\bm{I_n}$ preference lists;
    \Ensure
       \While {nonempty isn't empty}
       \State randomly selected to be placed in $nonempty$;
         \For{$n = 1$; $n<N$; $n++$ }
      \State  randomly assign the first $SC_{\rm f}$ in the preference list to $I_n$ ;
      \If {SC isn't matched }
      \State add a new entry $\left ( I_{\rm n},SC_{\rm f}\right ) $ to the list M;
      \State add one to the number of SC allocation to the corresponding cluster $I_n$;
      \State update nonempty;
       \State \textbf{break}
        \ElsIf {SC has been matched} 
        \State compare $m_{\rm prime}^{\rm rank}$ and $m^{\rm rank}$ in 
 SC preference lists
      \If {$m_{\rm prime}^{\rm rank}< m^{\rm rank}$}
      \State delete the SC from the I's preference list;
      \State \textbf{break}
       \ElsIf {$m_{\rm prime}^{\rm rank}> m^{\rm rank}$} 
        \State update M;
        \State delete SC from I's preference list.
\Else 
        \State \textbf{continue}
      \EndIf
      \EndIf
     \EndFor
    \EndWhile
  \end{algorithmic}
\end{algorithm}

\section{Alternate Iteration Algorithm for Joint Optimization}
\par Based on the optimization results in the above section, we propose an alternate optimization algorithm 2 to obtain the optimized solution of the original problem $P0$. In algorithm 2, the task allocation are first solved by given initial values. Then, the power allocation are solved by fixed other variables. The sub-channel matching and the bandwidth allocation are solved in the same way. By giving a limit to the optimization accuracy, iterative optimization is carried out to obtain the optimized computation efficiency of the system.
\par We consider setting $k$ as the iteration number, which represents the number of loops in Algorithm 2. The complexity of optimizing $\theta_{m}$ and $\zeta_{m}$ is $\mathcal{O}\left(1\right )$. It is a linear programming problem of a convex function. The optimization of $P_{m}$ and bandwidth allocation $\alpha_{m}$ adopts the method of quadratic transformation and conducts $M^2$ and $N^2$ iterations, respectively, with the computational complexity being $\mathcal{O}\left(M^2\right )$ and $\mathcal{O}\left(N^2\right )$. Algorithm 1 is used to optimize $\eta_{f,n}$, and its complexity is  $\mathcal{O}\left(G\right )$. Therefore, the total complexity of Algorithm 2 is $k\times \mathcal{O}\left(M^2+N^2\right )$.

\begin{algorithm}[h]
  \caption{Alternating Optimization Algorithm} 
  \begin{algorithmic}[1]
    \Require\\
    $\bm{\alpha}$, $\bm{P}$ and $\bm{\eta _{f,n}}$ to a feasible value;\\
    $k$: initial iteration number;
    \Ensure
       \While { $\left ( \mathbb{E}^{(\rm k)} -\mathbb{E}^{(\rm k-1)}\right)/\mathbb{E} ^{(\rm k)}>10^{\rm -5}$ \textbf{or}  $k\le50$} 
       \State Compute the optimal $\zeta _m^{\left ( \rm k \right )}$ and $\theta _m^{\left ( \rm k \right )}$ under given $ \alpha_m^{\left ( \rm k \right )}$,$P _m^{\left ( \rm k \right )}$ and $\eta _{f,n}^{\left ( \rm k \right )}$ by (\ref{32}) and (\ref{33});
       \State Compute the optimal $P_m^{\left ( \rm k \right )}$ under given $ \zeta_m^{\left ( \rm k \right )}$, $\theta _m^{\left ( \rm k \right )}$, $\alpha _m^{\left ( \rm k \right )}$ and $\eta _{f,n}^{\left ( \rm k \right )}$ by (\ref{44});
       \State Compute the optimal $\eta _{f,n}^{\left (\rm k \right )}$ under given $ \zeta_m^{\left ( \rm k \right )}$,$P _m^{\left ( \rm k \right )}$, $\alpha _m^{\left ( \rm k \right )}$ and $\eta _{f,n}^{\left ( \rm k \right )}$ by Algorithm 1;
       \State Compute the optimal $\alpha _m^{\left ( \rm k \right )}$  under given $ \theta _m^{\left ( \rm k \right )}$,$P _m^{\left ( \rm k \right )}$, $\zeta _m^{\left ( \rm k \right )}$ and $\eta _{f,n}^{\left ( \rm k \right )}$ by (\ref{38}) ;
       \State  $k=k+1$;
    \EndWhile
  \end{algorithmic}
\end{algorithm}

\section{Simulation Results}
In this section, we investigate the performance of the optimization scheme proposed in this paper by simulation. To ensure the accuracy and continuity of vehicular task offloading in STIN, we assume that each vehicle offloads its computational task to a single RSU located within a radius of $R=250$ m. This assumption is made to maintain reliable and efficient task offloading, as vehicles are expected to complete their computational tasks within the coverage area of one RSU. The radius of 250 meters is chosen to balance communication reliability and computation efficiency, ensuring that the vehicle can effectively offload its tasks without significant latency or disruption caused by frequent handovers between multiple RSUs. The total THz bandwidth for the transmission from vehicles to satellite is $B_s$=100 GHz \cite{10278635} \cite{9370130}.
\begin{table}[htbp]
\centering
\caption{Simulation parameters}
\begin{tabular}{ccc}
\hline\hline
Parameters                             &Symbol
&Values   \\
\hline
Size of task &$C_{\rm m}$ &$[3000-4200]\ kb$       \\
Maximum transmission power                  &$P^{\max}$   & 23 dBm           \\
Noise power                               &$\sigma ^{\rm 2}$    & -80 dBm          \\
RSU communication radius&$r_{\rm m}$ &250 m\\
Vehicles computational capabilities              &$Z_{\rm m} $               & 0.5 G cycles/s          \\
Path loss index                    &$\rho $          & 3.7          \\
\hline\hline
\end{tabular}
\label{tab:jianxie}
\end{table}
\par In Table I, we summarize some of the parameters involved in the simulation. The computation power of each BS, RSU and satellite can meet the demand of vehicle task transmission and computation. During the simulation, we compare the proposed jointly task offloading and resource allocation (JTORA) with different allocation methods.
\begin{itemize}
\item \textbf{Priority Locality}\\
The priority locality algorithm aims to allocate vehicle tasks to local vehicle for computation whenever possible under allowable conditions. For some simple tasks related to environmental perception data processing in vehicle autonomous driving, the computing units within the vehicle itself may be able to complete them quickly, and there is no need to transmit the data to edge devices for processing. By allocating tasks to local computing as much as possible, the time and  transmission consumption for transmitting data to other edge devices are reduced, the computational delay is effectively decreased, and thus the overall computation efficiency is improved. 
\item \textbf{Priority Edge}\\
The priority edge algorithm mainly assigns priority to edge devices for computing during the task allocation process. In the situation where vehicles generate a large amount of task and the computing capability of local devices is limited, the necessary task will be transmitted to edge devices for computing. In this way, the network bandwidth can be utilized more reasonably and the communication efficiency of the entire system can be improved. 
\item \textbf{Random Allocation}\\
The random allocation algorithm is a relatively simple and straightforward task allocation method. When dealing with the allocation of vehicle tasks to different computing devices, this algorithm does not consider complex factors such as the characteristics of tasks and the states of computing resources. Instead, it randomly assigns tasks to available computing devices. The greatest advantage of this algorithm lies in its simplicity. It does not require complex processes such as computing resource evaluation and task characteristic analysis, and is relatively easy to implement. Since it assigns tasks randomly and does not consider the adaptability between tasks and computing devices, it may lead to some unreasonable allocation situations. 
\item \textbf{One-to-One}\\
The one-to-one algorithm is a sub-channel matching algorithm in NOMA. It assigns one sub-channel to each vehicle, ensuring a one-to-one correspondence between sub-channels and vehicles. It is relatively simple to implement. The system does not require complex scheduling for multi-vehicle shared channels, thus reducing the complexity and computational overhead of the system. However, when the demands of users or the states of channels change, due to the fixed one-to-one matching, it is rather difficult to quickly and flexibly adjust the allocation of sub-channels. 
\item \textbf{Water Filling}\\
Water filling is an algorithm used for power allocation. More power is allocated to good channels, while less power is allocated to poor channels. By reasonably allocating power according to the channel state, the water filling algorithm can effectively improve the spectral efficiency. 
\item \textbf{Average Allocation}\\
The average allocation algorithm evenly distributes the total power among various vehicles. It is simple to implement, yet it fails to take into account the differences among vehicle, resulting in tasks with large resource requirements being unable to obtain sufficient resources and thus affecting the computation efficiency of the system. 
\end{itemize}

\subsection{Convergence of the Proposed Algorithm}
\begin{figure}[htbp]
\centerline{\includegraphics[scale=0.33]{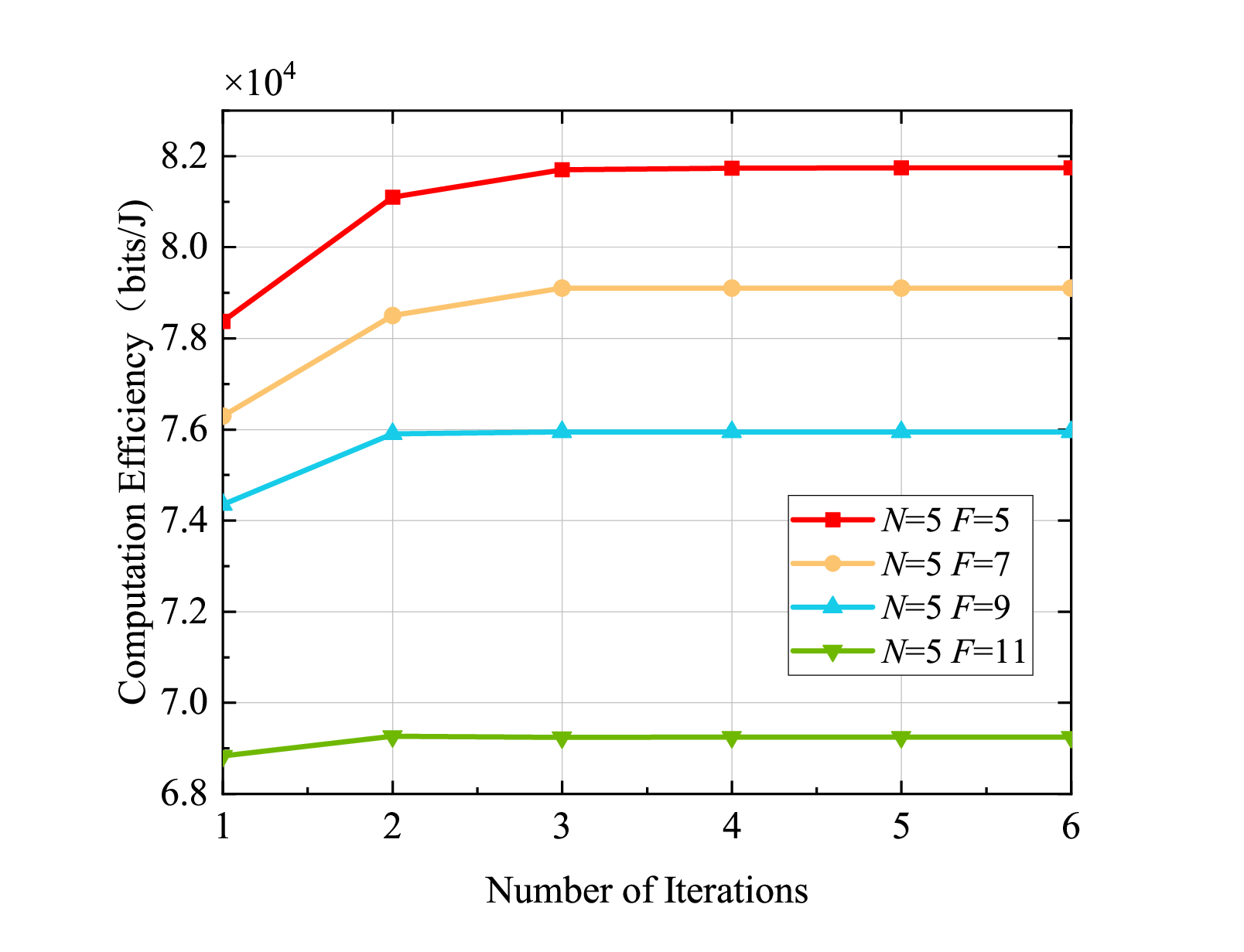}}
\caption{Iterations of VSMA under different sub-channels.}
\label{fig2}
\end{figure}
Fig. \ref{fig2} shows the comparison of the convergence performance of the proposed optimization algorithm under different numbers of sub-channels in the case of five vehicles in each NOMA group. It can be observed that as the number of sub-channels increases, the number of iterations remains almost the same, but the curve becomes flatter and the convergence speed is faster. With the number of vehicles fixed, more sub-channels can alleviate the competition for resources among vehicles. Each vehicle can be allocated relatively more appropriate sub-channels, reducing the frequent adjustments and iterations caused by resource competition, so that the matching process can be stabilized more quickly and convergence can be achieved.
\subsection{Performance Analysis}
\begin{figure}[htbp]
\centerline{\includegraphics[scale=0.33]{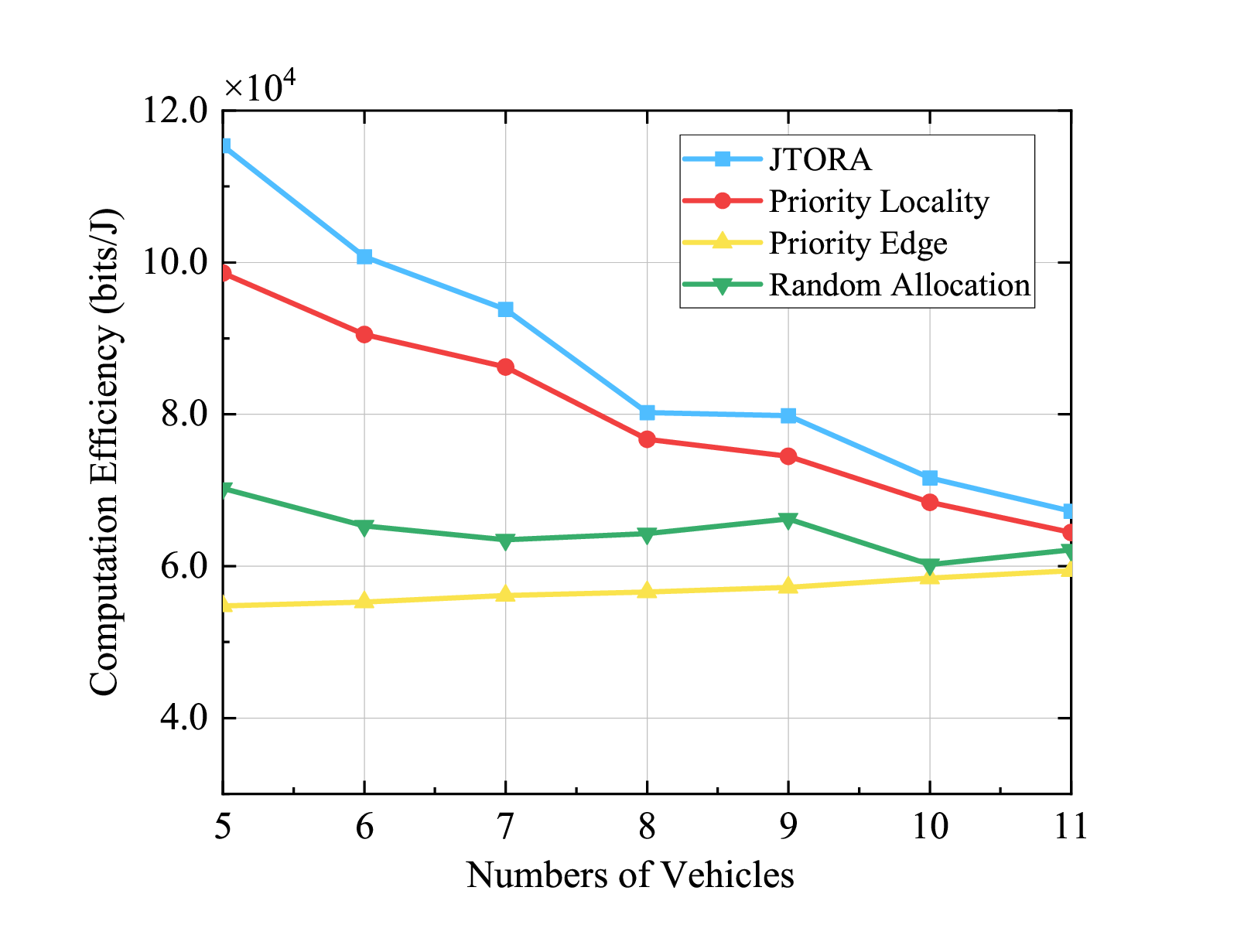}}
\caption{Computation efficiency vs the number of vehicles.}
\label{taskallocation vs. vehcle number}
\end{figure}
In Fig. \ref{taskallocation vs. vehcle number}, we compare the computation efficiency when the number of vehicles increases under different schemes of task allocation. It can clearly demonstrate that the computation efficiency decreases under different schemes as the number of vehicles increases. However, under the condition of prioritizing edge device computing, the computation efficiency is increasing. When the number of sub-channels remains fixed, with the increase of the number of vehicles, the sub-channel resources that can be allocated to each vehicle will decrease, resulting in a decrease in the rate of task data transmission and an increase in energy consumption.
Meanwhile, it can be observed that the proposed scheme performs better than prioritizing local computing. This is because local devices usually have relatively poor computing capabilities and will encounter numerous limitations when dealing with complex tasks. They are prone to performance bottlenecks when facing large-scale computing tasks, thereby resulting in additional time delays and energy consumption.
Moreover, the proposed scheme also has a better performance compared to random task allocation. In the proposed scheme, by means of multi-tier distributed computation, tasks are allocated to local devices, the BS and the LEO satellite, so that energy consumption can be reasonably distributed and unnecessary energy consumption can be reduced within a limited time.
In the computing prioritizing edge devices, as the number of vehicles increases, the task sizes of the system also increases. By allocating more tasks to BS with stronger computing capabilities, they can handle a large amount of data and complex computing tasks, thus effectively improving the computation efficiency.
\begin{figure}[htbp]
\centerline{\includegraphics[scale=0.33]{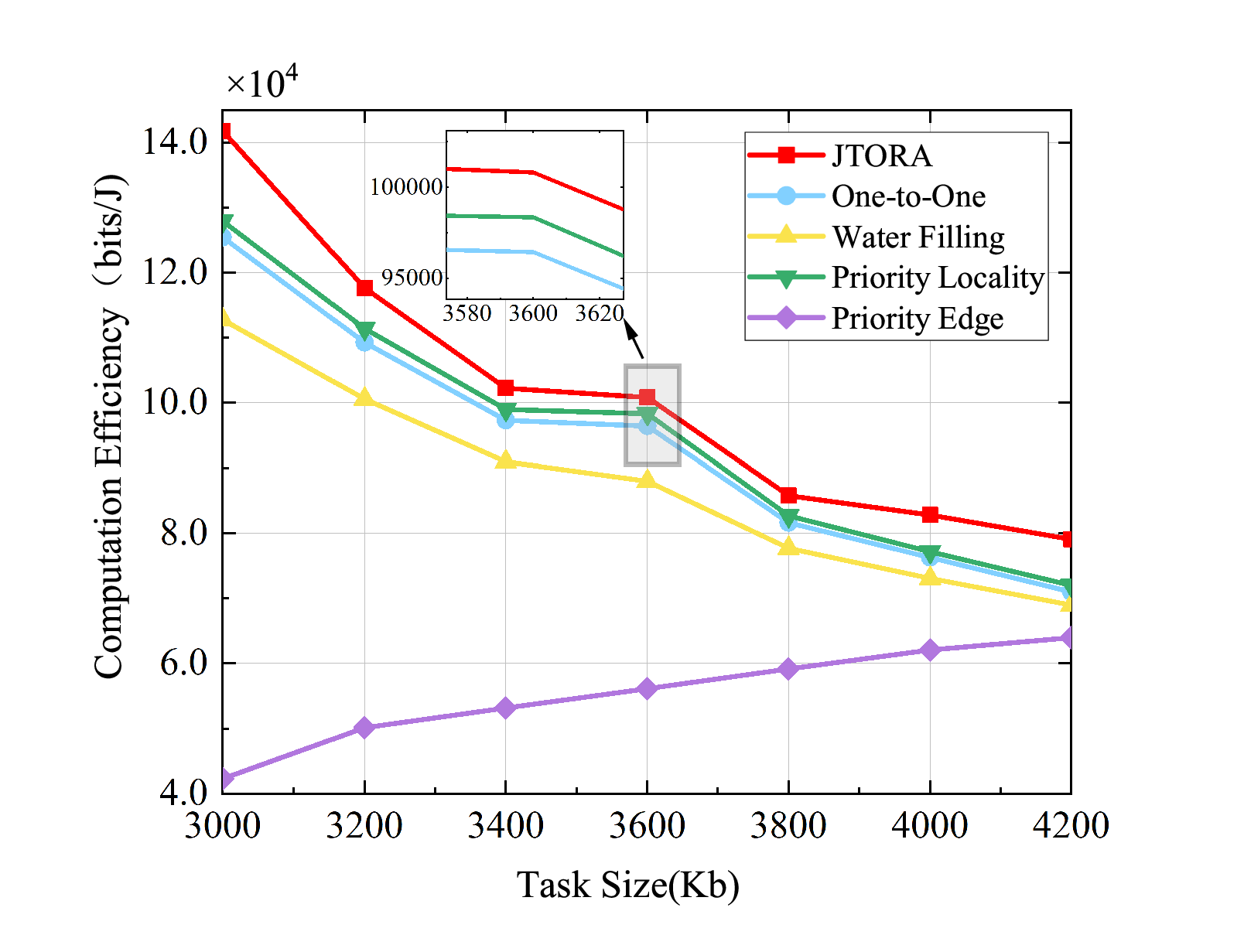}}
\caption{Computation efficiency versus task size.}
\label{task_size}
\end{figure}
\par In Fig. \ref{task_size}, we compare the computation efficiency when the task sizes increases under different schemes. It can be clearly observed that the computation efficiency decreases with the task sizes of vehicle increase under either of the schemes. But the computation efficiency is increasing under prioritizing edge device computing. The proposed optimization algorithm performs better than the one-to-one algorithm. This is because the many-to-one of the JTORA can dynamically allocate sub-channel resources according to the actual task transmission needs of vehicles. In contrast, in the one-to-one, the limited bandwidth restricts the transmission rate, increases energy consumption, and thus leads to a decline in computation efficiency. It may also cause sub-channels to be idle when some vehicles have no task to transmit, resulting in a waste of resources.\\
\par Meanwhile, the proposed optimization scheme also outperforms the performance of prioritizing local computing. The computing capacity of local cannot handle the increasingly large amount of task. The performance of water-filling algorithm for power allocation is also inferior to our proposed optimization scheme. Vehicles who have good channel quality may be allocated more power, while those located in the edge areas with poor channel quality may only be allocated very little power, or even cannot meet the basic communication requirements. Such unfair power allocation may lead to a serious decline in the communication quality of edge vehicles, such as an excessively low data transmission rate, and further increase energy consumption, reducing computation efficiency. For the computation that prioritizes edge devices, additional energy consumption is generated when tasks are transmitted to BS and the LEO satellite with greater computing capacity, thereby reducing the computation efficiency. 
\begin{figure}[htbp]
\centerline{\includegraphics[scale=0.33]{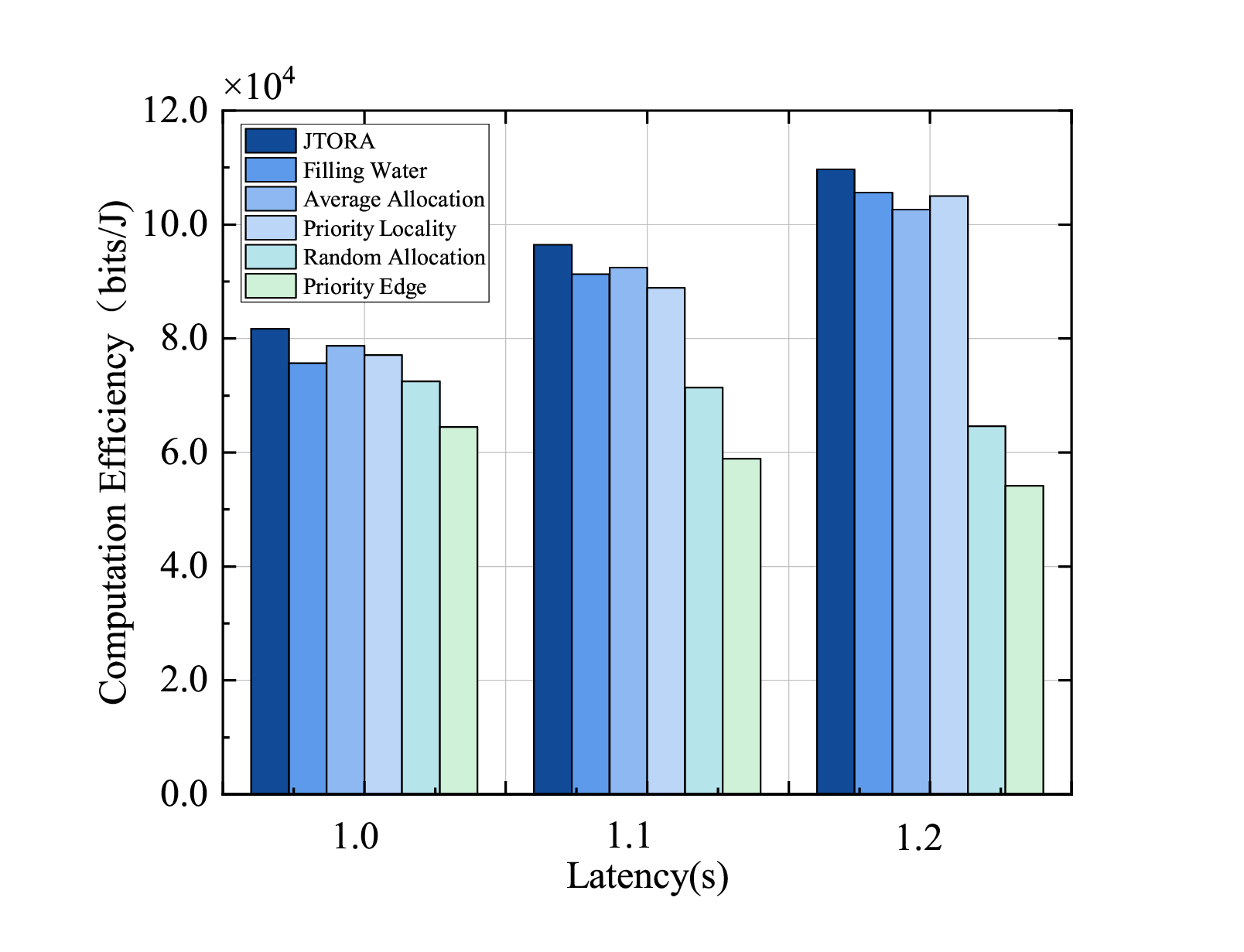}}
\caption{Computation efficiency versus maximum tolerable delay.}
\label{time}
\end{figure}
\par In Fig. \ref{time}, we compare the computation efficiency when the the maximum tolerable delay increases under different schemes. With the maximum tolerable delay increases, the computation efficiency increases. However, in the scheme of prioritizing edge device computing, the computation efficiency decreases. This is because as the maximum allowable delay increases, vehicles can choose to transmit more part of tasks to edge devices for computing within a defined range. JTORA improves the computation efficiency of the system by reasonably allocating tasks to local devices and edge devices. With the maximum tolerable delay increases, the computation efficiency increases. However, in the scheme of prioritizing edge device computing, the computation efficiency is decreasing. Because as the maximum tolerable delay increases, vehicles can choose to transmit tasks to edge devices for computing. JTORA improves the computation efficiency of the system by reasonably allocating tasks to local devices and edge devices. The unreasonable power allocations of the water-filling algorithm and the equal allocation of power have a great impact on the system energy consumption. If the equal power allocation is adopted, since the differences in vehicles' channel conditions are not taken into account, the system will be unable to fully utilize the channel advantages of vehicles with good channel quality to enhance the system capacity. When the equal power is allocated to vehicles with high channel gains and those with low channel gains, under the same bandwidth resources, vehicles with high channel gains could have transmitted more data by being allocated more power, thus improving the overall data transmission amount of the system, while the equal power allocation fails to achieve this. The equal power allocation does not allocate power reasonably according to vehicles' channel conditions and data transmission requirements, which may cause some vehicles to consume excessive energy in order to complete data transmission. As the maximum tolerable delay increases, more priority can be given to offloading tasks to edge devices for computing, which may lead to excessive transmission to BS and the LEO satellite for computing. This process affects energy consumption and reduces computation efficiency. 
\begin{figure}[htbp]
\centerline{\includegraphics[scale=0.33]{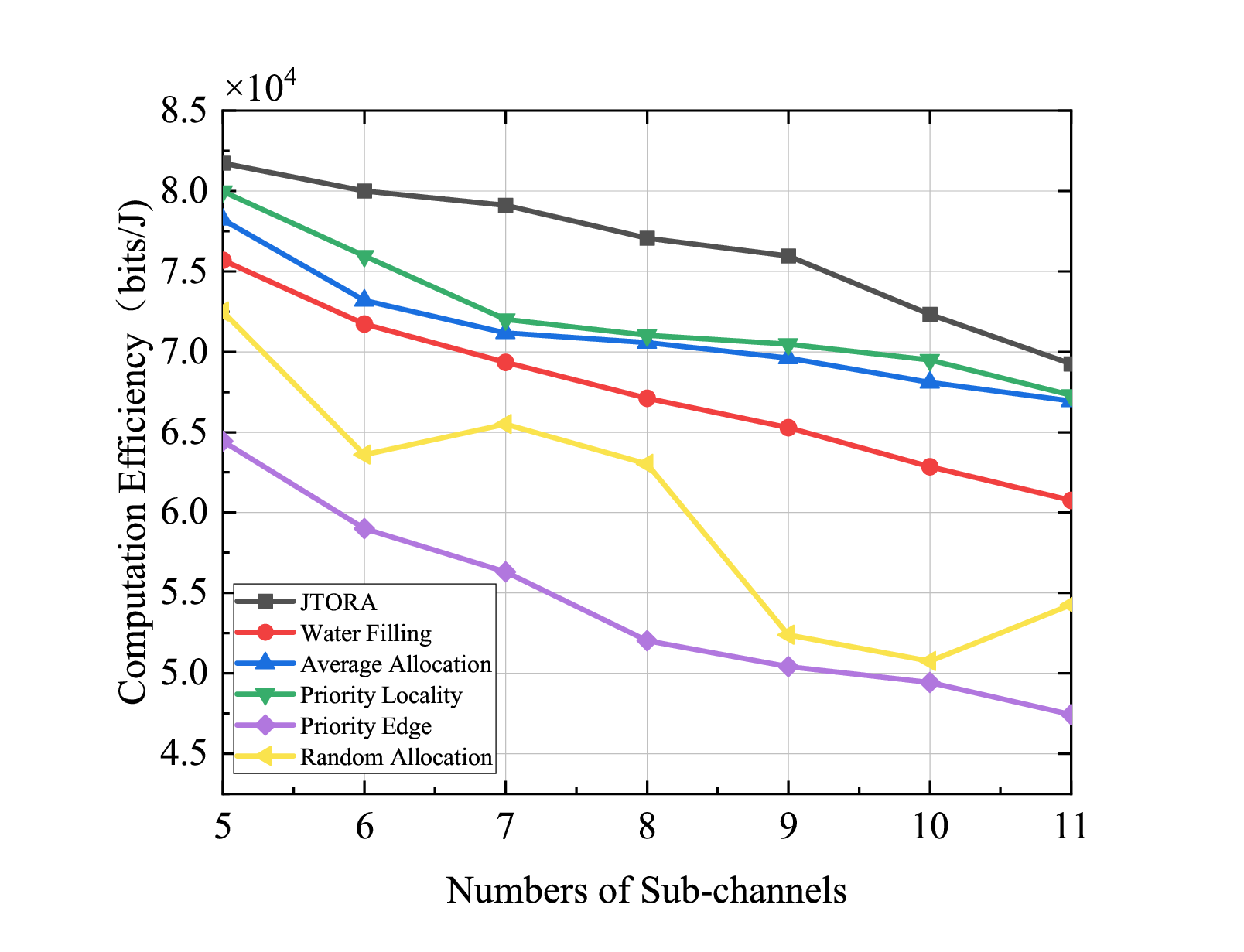}}
\caption{Computation efficiency versus the number of sub-channels.}
\label{subchannel}
\end{figure}
\par In Fig. \ref{subchannel}, we compare the computation efficiency when the the sub-channels increases under different schemes. With the number of sub-channels increases, the computation efficiency decreases. With the total bandwidth fixed in NOMA, as the number of sub-channels increases, the intra-cluster interference increases, and the transmission rate decreases accordingly. As a result, the system energy consumption increases and the computation efficiency decreases. Compared with the water-filling algorithm and the equal power allocation, JTORA can allocate more power to vehicles with better channel conditions according to the actual channel situations of vehicles, enabling them to complete data transmission with lower transmit power, thus reducing the energy consumption of the whole system and improving computation efficiency. However, the equal power allocation cannot achieve such an optimization effect. 

\begin{figure}[htbp]
\centerline{\includegraphics[scale=0.33]{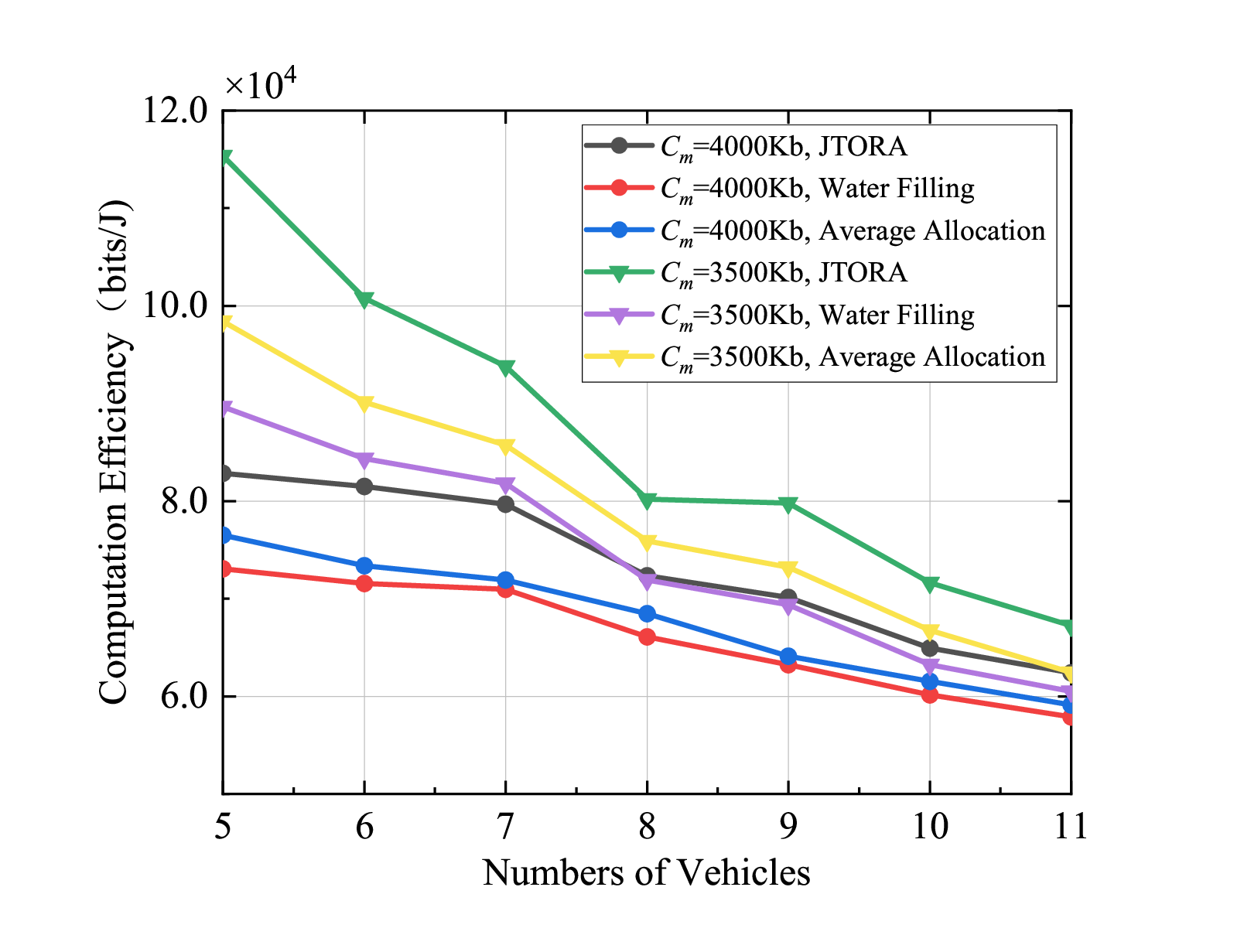}}
\caption{Computation efficiency with different allocation methods.}
\label{vehicle}
\end{figure}
\par In Fig. \ref{vehicle}, we compare the computation efficiency when the the number of vehicles increases under different schemes.  It can be clearly observed that the increase in the number of vehicles leads to a reduction in computation efficiency. Meanwhile, by setting the data volume sizes as $C_{n} = 3500$ Kb and $C_{n} = 4000$Kb, the computation efficiency shows a better performance when $C_{n} = 3500$Kb. The sub-channel resources that each vehicle can be allocated will decrease, resulting in a reduced transmission rate and an increased energy consumption. Meanwhile, the increase in the amount of task data also leads to the increase in computation energy consumption and computation delay. Our proposed JTORA scheme is significantly superior to the water-filling algorithm and the equal allocation, which is consistent with the previous analysis.

\section{Conclusion}
In this paper, we propose a novel VMDC framework with hybrid THz-RF transmission in STIN system. This system utilizes the OFDMA technique in the THz to achieve the optimal computation efficiency. For terrestrial task offloading, NOMA scheme is adopted. Then, we formulate a maximum computation efficiency optimization problem by jointly optimizing bandwidth allocation, vehicle task allocation, power allocation, and sub-channel matching in NOMA. Since the proposed problem are non-convex and coupled, in response to this challenge, we decouple the original problem into four sub-problems and propose an alternating optimization method. All independent variables are optimized and iterated until convergence to obtain the maximum computation efficiency. The simulation results show that the proposed strategy can achieve the best computation efficiency compared with the benchmark algorithms. In future work, we will consider a multi-tier distributed computing system with the collaboration of multiple LEO satellites.


%

\appendices
\section{Proof of the Theorem 1}
In the case of a large signal-to-noise ratio with infinite bandwidth, there are two limits:
\begin{equation}
\begin{aligned}
\lim_{B \to \infty} \frac{1}{x} Ib\left ( 1+x \right ) =Ibe=1.44.\label{52}
\end{aligned}
\end{equation}
\begin{equation}
\begin{aligned}
\lim_{B\to \infty} C=\lim_{B \to \infty} BIb\left ( 1+\frac{S}{n_0B}  \right )  .\label{53}
\end{aligned}
\end{equation}
According to (\ref{51}) and (\ref{52}), it is obtained
\begin{equation}
\begin{aligned}
\lim_{B \to \infty} C&=\lim_{B \to \infty} \left [  \frac{n_{0} B}{S}Ib\left ( 1+\frac{S}{n_{0} B}  \right )  \right ]\frac{S}{n_{0} }\\&=Ibe \frac{S}{n_0} =1.44\frac{S}{n_0}  .\label{54}
\end{aligned}
\end{equation}
Therefore, we can obtain that
\begin{align}
\overline{R} _{\rm m}=\frac{1.44P_{\rm m}\int_{0}^{t_{\rm m}^{\rm stay}} \sum_{f\in F} w_{f,k} \eta _{f,n}\frac{\left | h_{n,k}  \right | ^{2}d_{nR,m}^{- \rho ^{'} } }{n_0}dt }{t_{\rm m}^{\rm stay} } .
\end{align}

\section{Proof of the Lemma 1}
Suppose there exists a matching pair $\left ( \bm{I}_{k},SC_{f} \right ) $ in a matching $\mu^{\ast}$ which is $A\succ_{SC_{f} } \mu \left ( SC_{f}  \right )$, $A\subseteq \left \{ \bm{I}_{k}  \right \} \cup \left \{  SC_{f} \right \} $ and $\bm{I}_{k} \in A$ and $SC_{f} \succ_{\bm{I}_{k} } SC_{l}, SC_{l} \in \mu \left ( \bm{I}_{k} \right )$. In the algorithm 1, during the loop, unmatched clusters are randomly selected and an attempt is made to match them with sub-channels in the sub-channel's preference list. If the capacity of the selected sub-channel is not full, the cluster is directly matched with the sub-channel; if the sub-channel is full, the ranking of the new and matched clusters in the sub-channel's preference list is compared to decide whether to perform the replacement operation. Since the matching pair is random, the non-existence of the matching pair indicates that there is no blocking pair in that matching $\mu^{\ast}$. Therefore, matching of vehicles and sub-channels have finished successful in that matching $\mu^{\ast}$ and that matching $\mu^{\ast}$ is convergent and stable.

\section{Proof of the Theorem 2}
During each iteration, a partially unmatched cluster is randomly chosen from the cluster set and an attempt is made to allocate the corresponding sub-channel according to its preference. If the sub-channel has unfilled capacity, a match is established and the relevant state is updated. If the sub-channel is full, the rankings of the new cluster and the already matched clusters in the sub-channel's preference list are compared and adjusted accordingly. According to VSMA, since a swap operation makes the matching change from $\mu $ to $\mu^{\ast } $, $U_n\left ( f \right ) -U_n\left ( f -1\right ) > 0$ is valid. As a result, with the increase in the number of iterations, the set of available cluster matches will become smaller. Since the number of clusters and sub-channels is limited, the number of proposed match attempts will not exceed the total number of sub-channels. Therefore, the total number of iterations is finite, and the algorithm will conclude within a limited number of iterations and converge to a final stable match.



\ifCLASSOPTIONcaptionsoff
  \newpage
\fi

\bibliography{IEEEabrv,ref}
\end{document}